\documentclass[journal,10pt,draftcls,letterpaper,onecolumn]{IEEEtran}

\usepackage{cite}
\usepackage[T1]{fontenc}
\usepackage[latin9]{inputenc}
\usepackage{geometry}
\usepackage{verbatim}
\usepackage{float}
\usepackage{amsthm}
\usepackage[cmex10]{amsmath}
\usepackage{amssymb}
\usepackage{graphicx}
\usepackage{esint}
\usepackage{color,soul}

\makeatletter

\providecommand{\tabularnewline}{\\}
\floatstyle{ruled}
\newfloat{algorithm}{tbp}{loa}
\providecommand{\algorithmname}{Algorithm}
\floatname{algorithm}{\protect\algorithmname}

  \theoremstyle{definition}
  \newtheorem{defn}{\protect\definitionname}[section]
\theoremstyle{plain}
\newtheorem{thm}[defn]{\protect\theoremname}
  \theoremstyle{plain}
  \newtheorem{lem}[defn]{\protect\lemmaname}
  \theoremstyle{plain}
  \newtheorem{cor}[defn]{\protect\corollaryname}

\numberwithin{equation}{section} 

\@ifundefined{showcaptionsetup}{}{%
 \PassOptionsToPackage{caption=false}{subfig}}
\usepackage{subfig}
\makeatother

 \providecommand{\definitionname}{Definition}
 \providecommand{\lemmaname}{Lemma}
\providecommand{\corollaryname}{Corollary}
\providecommand{\theoremname}{Theorem}

\begin{document}

\title{Super-resolution on the Sphere using Convex Optimization}
\author{Tamir~Bendory,~
        Shai~Dekel,
        and~Arie~Feuer~\IEEEmembership{,Fellow,~IEEE} \\
        Revised: October 2014, January 2015}

\maketitle

\begin{abstract}
This paper considers the problem of recovering an ensemble of
Diracs on a sphere from its low resolution measurements. The Diracs can be located at any location on the sphere, not necessarily on a grid.
 We show that under a separation condition, one can recover the ensemble with high precision by a three-stage algorithm, which consists of solving a semi-definite program, root finding and least-square fitting. The algorithm's computation time depends solely on the number of measurements, and not on the required solution accuracy. 
  We also show that in the special case of non-negative ensembles, a sparsity condition is sufficient for recovery.
  Furthermore, in the discrete setting,  we estimate the recovery error in the presence of noise as a function of the noise level and the super-resolution factor. 
\end{abstract}

\section{Introduction}

In many cases, signals are observed on spherical manifolds. Typical
examples are astrophysics (e.g. \cite{komatsu2011seven,audet2011directional}),
gravity fields sensing \cite{klosko1982spherical} and geophysics
\cite{simons2006spatiospectral}. A further example is spherical microphone
arrays, used for spatial beam forming \cite{meyer2001beamforming}, sound recording \cite{meyer2003spherical} and acoustic source localization \cite{jarrett20103d}.

Spherical harmonics are a key tool for the analysis of signals on
the sphere. For instance, the spherical microphone array was analyzed
in terms of spherical harmonics in \cite{rafaely2005analysis}. Additionally,
spherical harmonics have been extensively used for various applications
in computer graphics, such as modeling of volumetric scattering effects,
bidirectional reflectance distribution function, and atmospheric scattering
(for more graphical applications, see \cite{sloan2008stupid} and
the references therein). Spherical harmonics are also used in medical
imaging \cite{taguchi2001cone,tournier2004direct,deslauriers2012spherical}, optical tomography
\cite{arridge1999optical}, wireless channel modeling \cite{pollock2003introducing}
and several applications in physics such as solving potential problem
in electrostatics \cite{macrobert1967spherical}, and the central
potential Schrodinger equation in quantum mechanics \cite{cohen2006quantum}.
Based on spherical harmonics analysis, new sampling theorems on the
sphere for band-limited signals \cite{mcewen2011novel,ben2012generalized} and  for signals with finite rate of innovation \cite{deslauriers2013sampling} were suggested, and advanced analysis methods on the sphere were applied \cite{khalid2013spatially,leistedt2012exact}.

Let $\mathcal{H}_{n}({\mathbb{S}^{d-1}})$ denote the space of homogeneous
spherical harmonics of degree $n$, which is the restriction to the
$(d-1)$ unit sphere of the homogeneous harmonic polynomials of degree
$n$ in $\mathbb{R}^{d}$ \cite{atkinson2012spherical}. Each subspace
$\mathcal{H}_{n}({\mathbb{S}^{d-1}})$ is of dimension 
\[
a_{n,d}:=\frac{(2n+d-2)(n+d-3)!}{n!(d-2)!},\quad n\in\mathbb{N},d\geq2.
\]

Let us denote by $\{Y_{n,k}\}$, $k=1,...,a_{n,d}$, an orthonormal
basis of $\mathcal{H}_{n}({\mathbb{S}^{d-1}})$. The set $\{Y_{n,k}\}$
is a basis for the space of square integrable functions on $\mathbb{S}^{d-1}$.
Consequently, any $f\in L_{2}(\mathbb{S}^{d-1})$ can be expanded
as 
\begin{align}
f & =\sum_{n=0}^{\infty}\sum_{k=1}^{a_{n,d}}\langle f,Y_{n,k}\rangle Y_{n,k}.\label{eq:f_exp}
\end{align}
For $d=2$ the appropriate spherical harmonic basis is simply the
standard Fourier basis $\left\{ e^{jn\phi}\right\} $.

In this work we focus on the two-dimensional sphere $\mathbb{S}^{2}$
embedded in $\mathbb{R}^{3}.$ In this case, any point on the sphere
is parametrized by $\xi:=\left(\phi,\theta\right)\in\left[0,2\pi\right)\times\left[0,\pi\right].$
The appropriate orthonormal spherical harmonics basis is given by%
\footnote{Note that $k$ has a different range here than in (\ref{eq:f_exp}). %
} 
\begin{equation*}
Y_{n,k}\left(\xi\right)=A_{n,k}e^{jk\phi}P_{n,k}\left(\cos\theta\right),\quad0\le n<\infty,\quad-n\le k\le n,\label{eq:SH_2d}
\end{equation*}
where $P_{n,k}\left(x\right)$ is an \emph{associated Legendre polynomial
}of degree $n$ and order $k$, and $A_{n,k}$ is a normalization
factor, given by 
\begin{equation}
A_{n,k}:=\sqrt{\frac{2n+1}{4\pi}\frac{\left(n-\left|k\right|\right)!}{\left(n+\left|k\right|\right)!}}.\label{eq:A}
\end{equation}
The distance on the sphere between any two points $\xi_{i},\xi_{j}\in\mathbb{S}^{2}$
is given by 
\[
d(\xi_{i},\xi_{j})=\arccos\left(\xi_{i}\cdot\xi_{j}\right).
\] 
Consider a Dirac ensemble on the bivariate sphere $\mathbb{S}^{2}$
\begin{equation}
f=\sum_{m}c_{m}\delta_{\xi_{m}},\label{eq:signal}
\end{equation}
where $\delta_{\xi}$ is a Dirac measure, $\left\{ c_{m} \right\}$ are real weights,
and $\Xi:=\left\{ \xi_{m}\right\} \subset\mathbb{S}^{2}$ are distinct
locations on the sphere, namely the signal support.
Let us denote by $V_{N}$ the space of spherical harmonics of degree
$\leq N$. We assume that the only information we have on the signal
$f$ is its 'orthogonal projection' onto $V_{N}$, i.e. 
\begin{equation}
y_{n,k}:=\langle f,Y_{n,k}\rangle=\sum_{m}c_{m}\overline{Y}_{n,k}(\xi_{m}),\quad0\le n\le N,\quad-n\le k\le n.\label{eq:meas}
\end{equation}
In matrix notations, (\ref{eq:meas}) is presented as 
\begin{equation}
y=F_{N}f,\label{eq:matrix}
\end{equation}
where $F_{N}$ is a semi-infinite matrix with $\left(N+1\right)^{2}$
rows and $y$ is a column stacked vector of $\left\{ y_{n,k}\right\} $.
That is to say, $F_{N}$ is a projection operator onto $V_{N}$. The
adjoint operator is denoted as $F_{N}^{*}$. Our first main contribution (see Theorem \ref{Th:main}) is an algorithm that recovers exactly the underlying signal from its projection onto $V_{N}$.

To be clear, we assume that the high spherical harmonic coefficients are annihilated before any sampling procedure occurs.
In the spatial domain, the projection onto $V_N$ can be computed by approximately $2(N+1)^2$ samples based on a stable equiangular sampling scheme on the sphere \cite{mcewen2011novel}. A recent work derives an accurate computation of the projection using only $\left(N+1\right)^2$ samples \cite{khalid2014optimal}.

As a special case of the analog model, we also define a discrete configuration where the signal is known to lie on a grid.
Consider a discrete signal on the sphere 
\begin{equation}
f=\sum_{m}c_{m}\delta_{\xi_{m}},\quad\xi_{m}\in\Xi\subset\mathbb{S}_{L}^{2},\label{eq:signal-dis}
\end{equation}
where $\mathbb{S}_{L}^{2}$ is a predefined grid, not necessarily uniform.  
We assume that  any pair of points on the grid $\xi_i,\xi_j\in \mathbb{S}_{L}^{2}$ obey $d(\xi_i,\xi_j)\geq 1/L$ for some  $L\geq 1/\pi$.
This measurements model is equivalent
to 
\[
y=F_{N}^{L}f,
\]
where $F_{N}^{L}$ is the spherical harmonics matrix. This model will serve as the basis for our main result on recovery in noisy setting (see Theorem \ref{Theorem2}).
For the discrete model, {we define the notion of super-resolution factor} ($SRF$) (see also \cite{candes2013towards}).
SRF is  defined as 
\begin{equation} \label{eq:SRF}
SRF:=\frac{L}{N},
\end{equation}
and represents the ratio between the desired and the measured resolutions.
This agrees with the analog model (\ref{eq:signal}) when $SRF\rightarrow\infty.$

Our model reflects the fact that sensing systems have a physical
limit, determining the highest resolution the system can achieve.
In these cases, the observer has access solely to a coarse scale measurements
of the underlying signal. The problem of recovering the fine details
of a signal from its low-resolution measurements can be interpreted
as super-resolution on the sphere problem.

This work was inspired by the seminal paper of Candes and Fernandez-Granada
\cite{candes2013towards}, who investigated the recovery of Dirac
ensemble on the interval $[-\pi,\pi]$ from its low $2N+1$ Fourier
coefficients. The main result of this paper states that if the Diracs
are separated by at least $\frac{4\pi}{N}$, the signal can be recovered
as the unique solution of a tractable convex optimization problem.
This result holds for higher dimensions as well under a separation
condition of ${C_{d}}/{N}$, where $C_{d}$ is a constant which
depends only on the dimension of the problem (e.g. $C_{1}=4\pi$).
A consecutive paper \cite{candes2013super} showed that the recovery
is robust to noisy measurements. Similar results are given for support
detection from low Fourier coefficients \cite{fernandez2013support,azais2013spike}, recovery of non-uniform splines from their projection onto
spaces of algebraic polynomials \cite{bendory2013Legendre,de2014non} and recovery of streams of pulses \cite{sop,SOP_US}. 
(see also \cite{de2012exact}).

The configuration in (\ref{eq:matrix}) resembles the formulation
in compressed sensing (CS) (e.g. \cite{donoho2006compressed,candes2006robust}).
Using CS methods, the authors of \cite{rauhut2011sparse} have suggested
to recover a  s-sparse signal with bandwidth $N$ (in the sperical harmonics domain) by only $m\sim sN^{1/2}\log^4\left(N^2\right)$ samples using $\ell_1$  minimization.  In \cite{burq2012weighted}, the number of the required samples was reduced to  $m\sim sN^{1/3}\log^4\left(N^2\right)$.
However, we note that there exist two important distinctions between
the framework suggested here and CS. Firstly, CS usually works on discrete
signals, while (\ref{eq:signal}) describes an analog model, namely
the support $\Xi$ can comprise any point on the sphere. Secondly,
CS sampling matrix is required to be incoherent in some sense, which
typically leads to random sampling strategies, while in (\ref{eq:meas})
the measurements consist of the low-end of the spherical harmonics
representation, and as a result are extremely coherent.

Recently, a number of works suggest to super-resolve signals by a
semi-definite program \cite{candes2013towards,tang2012compressive,bhaskar2011atomic,tang2013near,chi2013compressive,xu2013precise}.
We extend this line of work to signals on a sphere. The first result
of this paper is that Algorithm \ref{alg1} recovers a signal of the
form (\ref{eq:signal}) from its low-resolution measurements (\ref{eq:meas}) using a three-stage algorithm consists of solving a semi-definite program, root findind on the sphere, and least square fitting. This holds provided that the Diracs are separated by at least $\frac{\nu}{N}$ for some numerical constant $\nu$.

In \cite{deslauriers2013sampling}, the authors suggest a parametric method ('finite rate of innovation' type) to reconstruct exactly a stream of K Diracs on the sphere from $3K$ samples, which is optimal (that is to say, the number of samples is equal to the number of degrees of freedom). This approach assumes a known number of Diracs, but does not assume any separation between the Diracs. Generally, parametric methods such as MUSIC, matrix pencil and ESPRIT \cite{stoica2005spectral, 56027, roy1989esprit,schmidt1986multiple } tend to be unstable in the presence of noise or model mismatch. Our second result generalizes \cite{candes2013towards} to the sphere in the discrete setting (\ref{eq:signal-dis}) and provides an estimate of the recovery error in the presence of noise or model mismatch. 

The rest of the paper is organized as follows. Section \ref{sec:main_results} presents the two main results of this paper, and sections \ref{sec:Proof-of-the} and \ref{sec:proof_noise} prove them. Section \ref{sec:Numerical-Experiments}
is devoted to experimental results.
 Section \ref{sec:necessity} elaborates on the
necessity of the separation condition and ultimately Section \ref{sec:Conclusions-and-Future}
concludes the paper and relates it to an ongoing
research.

\section{Main Results} \label{sec:main_results}

In a previous paper \cite{bendory2013exact}, the authors established
a sufficient condition for exact recovery of a signal of the form
of (\ref{eq:signal}) from its projection onto $V_{N}$ using a convex
optimization method. The recovery relies on the following separation
condition:
\begin{defn}
\label{def:separation}A set of points $\Xi\subset\mathbb{S}^{2}$
is said to satisfy the minimal separation condition for (sufficiently
large) $N$ if 
\[
\Delta:=\min_{\xi_{i},\xi_{j}\in\Xi,\xi_{i}\neq\xi_{j}}d\left(\xi_{i},\xi_{j}\right)\geq\frac{\nu}{N},
\]
where $\nu$ is a fixed constant that does not depend on $N$. 
\end{defn}
Under the separation condition, the points $\xi_m\in\Xi$ are center of pairwise disjoint caps of area $2\pi\left(1-\cos\left(\frac{\nu}{2N}\right)\right)$\cite{atkinson2012spherical}. Consequently, the number of points on the sphere can be roughly estimated by $\frac{2}{1-\cos\left(\frac{\nu}{2N}\right)}$ (for a tighter estimation, see \cite{rankin1955closest}). In an noise-free environment, a separation constant of $2\pi$ seems to ensure exact recovery (see Figure \ref{fig:separation}). This separation coincides with the spatial resolution of the projection of $f$ onto $V_N$, namely  $F_N^*F_Nf:=P_Nf$  \cite{rafaely2004plane}.
  In a noisy environment, we increased the separation constant to be $\nu=2.5\pi$.

Before presenting the main theorem, we introduce the notion of half
space. A half space is a set $\mathcal{H\subset\mathbb{Z}}^{d}$,
satisfying $\mathcal{H}\cap\left(-\mathcal{H}\right)=\left\{ 0\right\} $,
$\mathcal{H}\cup\left(-\mathcal{H}\right)=\mathbb{Z}^{d}$, and $\mathcal{H}+\mathcal{H}\subset\mathcal{H}$
\cite{dumitrescu2007positive}. Figure \ref{fig:Half-space-in-1}
demonstrates the two half spaces in two dimensions.

We make use of the following notations. Let $\Theta_{k}$ be an elementary
Toeplitz matrix with ones on the $k$ diagonal and zeros elsewhere
(the main diagonal is indexed by zero), let $\otimes$ be a Kronecker
product and $\Theta_{k,\ell}:=\Theta_{\ell}\otimes\Theta_{k}$. $tr\left(X\right)$
denotes the trace of the matrix $X$, and $\delta_{k\text{,\ensuremath{\ell}}}$
denotes a Kronecker Delta function, defined as 
\begin{equation}
\delta_{k\text{,\ensuremath{\ell}}}=\begin{cases}
1 & k=\ell=0,\\
0 & othewise.
\end{cases}\label{eq:delta}
\end{equation}

\begin{figure}[h]
\begin{centering}
\includegraphics[scale=0.5]{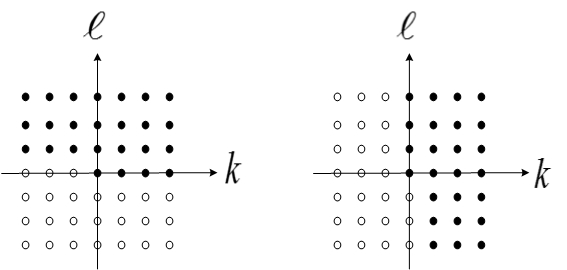} 
\par\end{centering}

\protect\caption{Half spaces in two dimensions. \label{fig:Half-space-in-1}}
\end{figure}

\begin{algorithm}
\textbf{Input:} The signal's projection onto $V_{N}$ (\ref{eq:meas}).

\textbf{Output: }A signal of the form (\ref{eq:signal}). 
\begin{enumerate}
\item Solve the semi-definite program 
\begin{equation} \label{eq:opt}
max_{\alpha,Q}\left\langle y,\alpha\right\rangle -\varepsilon\Vert\alpha\Vert_{2}\quad\mbox{subject to}\quad\begin{bmatrix}Q & \grave{h}\\
\grave{h}^{*} & 1
\end{bmatrix}\succeq0,\quad tr\left(\Theta_{k,\ell}Q\right)=\delta_{k\text{,\ensuremath{\ell}}},\quad(k,\ell)\in\mathcal{H},
\end{equation}
where $Q\in\mathbb{C}^{(2N+1)^{2}\times(2N+1)^{2}}$ is a Hermitian
matrix, $\grave{h}\in\mathbb{C}^{(2N+1)^{2}}$ is related to $\alpha$
through (\ref{eq:h}), and $\mathcal{H}$ is a half space. 
\item Define $q(\xi)=F_{N}^{*}\alpha(\xi)$, and find the roots of the polynomials
$1-q\left(\xi\right)$ and $1+q\left(\xi\right)$. These roots are
denoted as $\hat{\Xi}:=\left\{ \hat{\xi}_{m}\right\} .$ 
\item Solve the least-square system 
\[
\sum_{\hat{\xi}_{m}\in\hat{\Xi}}\hat{c}_{m}\overline{Y}_{n,k}\left(\hat{\xi}_{m}\right)=y_{n,k},\quad n=0,\dots,N,\thinspace k=-n,\dots,n.
\]

\item Construct the recovered signal as 
\[
\hat{f}=\sum_{m}\hat{c}_{m}\delta_{\hat{\xi}_{m}}.
\]

\end{enumerate}
\protect\caption{Recovery of a signal of the form (\ref{eq:signal}) from its projection
onto the space of spherical harmonics of degree $\leq N$. \label{alg1}}
\end{algorithm}

Algorithm \ref{alg1} consists of three stages: solving a semi-definite
program, root finding and least-square fitting. Although the model
(\ref{eq:signal}) reflects an analog (infinite dimensional) signal,
we suggest to recover it from (\ref{eq:meas}) by a semi-definite
program with $\mathcal{\ensuremath{O}}\left(N^{4}\right)$ variables.
This results in an unconstrained accuracy and no dependence on any
discretization step. 
\begin{thm}
\label{Th:main}Let $\Xi=\{\xi_{m}\}$ be the support of a signed
measure f of the form $(\ref{eq:signal})$. Let $\{Y_{n,k}\}_{n=0}^{N}$
be any spherical harmonics basis for $V_{N}(\mathbb{S}^{2})$ and
let $y_{n,k}=\langle f,Y_{n,k}\rangle$, $0\le n\le N$, $-n\le k\le n$.
If $\Xi$ satisfies the separation condition of Definition $\ref{def:separation}$,
then Algorithm \ref{alg1} recovers f exactly with $\varepsilon=0$ in (\ref{eq:opt}).
Furthermore, Algorithm \ref{alg1} recovers a non-negative signal
\emph{(}i.e.\emph{ $c_{m}>0$)} exactly as long as $f$ has at most
$N$ non-zero values. 
\end{thm}

In addressing the noisy case, we consider the following discrete model. Let us denote by $D_L$ the set 
of measures of the form (\ref{eq:signal-dis}), that is, $f=\sum_m c_m\delta_{\xi_m}$, with $\xi_m \in \mathbb{S}_L^2$, $\forall m$. 
Observe that $f$ can be regarded as a discrete signal $\left\{ c_m \right\}$ indexed by the set $\mathbb{S}_L^2$.
Therefore, we may also define for
$f \in D_L$, $\Vert f\Vert_{\ell_p}:=\left({\sum_m\vert c_m\vert^p}\right)^{1/p},\thinspace p\geq 1$.  Note that for $f \in D_L$, 
$\left\Vert f \right\Vert_{TV}=\left\Vert f \right\Vert_{\ell_1}$ (see Section \ref{sec:Proof-of-the} for definition of the TV norm for measures). Next, we consider noisy input data of the type 
\begin{equation}
y_{n,k}:=\langle f,Y_{n,k}\rangle+\eta_{n,k}=\sum_{m}c_{m}\overline{Y}_{n,k}(\xi_{m})+\eta_{n,k},\quad 0\le n\le N,\quad-n\le k\le n,\label{5}
\end{equation}
where $f\in D_L$, where $\eta_{n,k}$ is an additive noise.

 The following result shows that the recovery error using $\ell_1$ minimization is proportional to $SRF^2$ and the noise level.
\begin{thm}
\label{Theorem2} Let $\Xi=\{\xi_{m}\}\subset \mathbb{S}_L^2$ be the support of a signed
measure $f\in D_L$ (i.e. of the form $(\ref{eq:signal-dis})$), where $\eta :=\left\{\eta_{n,k}\right\}$ satisfies $\Vert \eta\Vert_2\leq \varepsilon$. Let $\{Y_{n,k}\}_{n=0}^{N}$
be any spherical harmonics orthobasis for $V_{N}(\mathbb{S}^{2})$ and
let $\{y_{n,k}\}$ be as in (\ref{5}).
For sufficiently large $L$, if $\Xi$ satisfies the separation condition
of Definition $\ref{def:separation}$, then the solution $\hat{f}$ of 
\begin{equation} 
\min_{g\in D_L}\left\Vert {g}\right\Vert _{\ell_1}\quad\mbox{subject to}\quad\left\Vert y-F_{N}g\right\Vert _{\ell_2}\leq \varepsilon,\label{6}
\end{equation}
satisfies
\[
\Vert\hat{f}-f\Vert_{\ell_1}\leq C_{0}SRF^{2}\varepsilon,
\]
where $C_0$ is a numerical constant. 
\end{thm}

We have chosen to work with a bounded noise, however our technique can be extended to other noise models. For instance, suppose that $\eta_{n,k}$ are iid entries $\mathcal{N}\sim(0,\sigma^{2})$. In this case we obtain the following corollary:
\begin{cor}
\label{Corollary1} Consider the model (\ref{5}) and suppose that $\eta_{n,k}$ are iid entries $\mathcal{N}\sim(0,\sigma^{2})$. Fix $\gamma>0$. For sufficiently large $L$, if $\Xi$ satisfies the separation condition
of Definition $\ref{def:separation}$, then the solution $\hat{f}$ of (\ref{6}) with $\varepsilon=(N+1)(1+\gamma)\sigma$
satisfies
\[
\Vert\hat{f}-f\Vert_{\ell_1}\leq C_{0}(1+N)(1+\gamma)SRF^{2}\sigma,
\]
with probability of at least $1-e^{-\frac{1}{2}\left(N+1\right)^{2}\gamma^{2}}$,
where $C_0$ is a numerical constant. 
\end{cor}
Theorem \ref{Theorem2} and Corollary \ref{Corollary1} are proved in Section \ref{sec:proof_noise}.
 
\section{Proof of Theorem \ref{Th:main} \label{sec:Proof-of-the}}

The proof of Theorem \ref{Th:main} relies on a few results from \cite{bendory2013exact}.
To this end, recall the following definition \cite{bendory2013exact,rudin1986real}:
\begin{defn}
Let $\mathcal{B}(A)$ be the Borel $\sigma$-Algebra on a compact
space $A$, and denote by $\mathcal{M}(A)$ the associated space of
real Borel measures. The Total Variation of a real Borel measure $v\in\mathcal{M}(A)$
over a set $B\in\mathcal{B}(A)$ is defined by 
\[
\vert v\vert(B)=\sup\sum_{k}\vert v(B_{k})\vert,
\]
where the supremum is taken over all partitions of $B$ into a finite
number of disjoint measurable subsets. The total variation $\vert v\vert$
is a non-negative measure on $\mathcal{B}(A)$, and the Total Variation
(TV) norm of $v$ is defined as 
\[
\|v\|_{TV}=\vert v\vert(A).
\]
 
\end{defn}
In short, the total variation norm of a signed measure can be interpreted
as the generalization of $\ell_{1}$ norm to the real line.
This is \emph{not} the total variation of a function, a frequently-used regularizer in signal processing (see \cite{mcewen2013sparse} for the definition of the discrete total variation on the sphere). For a
measure of the form of (\ref{eq:signal}), it is easy to see that
\[
\|f\|_{TV}=\sum_{m}\vert c_{m}\vert.
\]

The following lemma concerns the existence of an interpolating polynomial
as follows \cite{bendory2013exact}:
\begin{lem}
\label{Lemma:q}If $\Xi$ satisfies the separation condition of Definition
\ref{def:separation}, then there exists a polynomial $q\in V_{N}$
such that 
\begin{eqnarray*}
q(\xi_{m}) & = & u_{m},\quad\forall\xi_{m}\in\Xi,\\
q(\xi) & < & 1,\quad\xi\notin\Xi,
\end{eqnarray*}
for any signed set $\{u_{m}\}$ with $|u_{m}|=1.$
\end{lem}
The main Theorem of \cite{bendory2013exact} is the following:
\begin{thm}
\label{Th:main1}Let $\Xi=\{\xi_{m}\}$ be the support of a signed
measure f of the form $(\ref{eq:signal})$. Let $\{Y_{n,k}\}_{n=0}^{N}$
be any spherical harmonics basis for $V_{N}(\mathbb{S}^{2})$ and
let $y_{n,k}=\langle f,Y_{n,k}\rangle$, $0\le n\le N$, $-n\le k\le n$.
If $\Xi$ satisfies the separation condition of Definition $\ref{def:separation}$,
then $f$ is the unique solution of
\begin{equation}
\min_{g\in\mathcal{M}(\mathbb{S}^{2})}\|g\|_{TV}\quad\mbox{subject to}\quad F_{N}g=y,\label{eq:tv}
\end{equation}
 where $\mathcal{M}(\mathbb{S}^{2})$ is the space of signed Borel
measures on $\mathbb{S}^{2}$. 
\end{thm}
Theorem \ref{Th:main1} states that if the underlying signal
satisfies the separation condition of Definition \ref{def:separation},
then the signal is the unique solution of the TV minimization (\ref{eq:tv}).
Furthermore, in the case of non-negative signals it has been shown
that the solution of (\ref{eq:tv}) is precise as long as the signal
has at most $N$ non-zero values, that is the separation condition may be replaced
by a weaker sparsity condition \cite{bendory2013exact}. 

The challenge of solving (\ref{eq:tv}) is that the analog nature
of the signal dictates an infinite-dimensional problem. One approach
to alleviate this problem is to assume that the signal lies on a grid.
In this case, the TV minimization reduces to standard $\ell_{1}$
minimization. However, the discretization leads to an unavoidable
error, which can be mitigated by refining the grid, which in turn
increases the problem complexity. This case is analyzed in Section
\ref{sec:Numerical-Experiments}. In contrast, we suggest a different approach
with (theoretically) unlimited accuracy and no dependence on discretization step.

The algorithm consists of three steps. First, we reformulate the dual
problem of TV minimization as a finite semi-definite programming.
Later on, we use the dual solution to locate the signal's support
by root finding. Finally, we estimate the amplitudes (i.e. the weights
$c_{m}$) by least-square estimation. 

With the notation $\alpha:=\left\{ \alpha_{n,k}\right\}$, let $F_{N}^{*}\alpha\left(\xi\right):=\sum_{0\leq n\leq N,k}\alpha_{n,k}Y_{n,k}\left(\xi\right)$.
We assume that $\alpha$ belongs to the subspace of vectors for which
\[
\left\Vert F_{N}^{*}\alpha\right\Vert _{\infty}:=\max_{\xi\in S^{2}}\left|\sum_{n,k}\alpha_{n,k}Y_{n,k}\left(\xi\right)\right|\leq1.
\]
Under the separation condition, this give us the following duality 
\begin{eqnarray*}
max_{\alpha}\left\langle y,\alpha\right\rangle  & := & max_{\alpha}\left\langle F_{N}f,\alpha\right\rangle =max_{\alpha}\left\langle f,F_{N}^{*}\alpha\right\rangle \\
 & = & max_{\alpha}\int_{S^{2}}F_{N}^{*}\alpha\left(\xi\right)df\left(\xi\right)\\
 & = & \int_{S^{2}}q\left(\xi\right)df\left(\xi\right)=\left\Vert f\right\Vert _{TV},
\end{eqnarray*}
where $q$ is the polynomial from Lemma \ref{Lemma:q}.
Therefore, (\ref{eq:tv}) is translated to
\begin{equation}
max_{\alpha}\left\langle y,\alpha\right\rangle \quad s.t.\quad\left\Vert F_{N}^{*}\alpha\right\Vert _{\infty}\leq1.\label{eq:opt1}
\end{equation}
Observe that the objective function is finite dimensional, whereas
the constraint is of infinite dimension. To this end, we state the
following variant of the Bounded Real Lemma \cite{dumitrescu2007positive}:
\begin{lem}
Consider a causal trigonometric polynomial of the form 
\[
H(\omega_{1},\omega_{2})=\sum_{k,\ell=0}^{N}h_{k,\ell}e^{-j\left(\omega_{1}k+\omega_{2}\ell\right)}.
\]
The following inequality holds 
\[
\left|H\left(\mathbf{\omega_{1},\omega_{2}}\right)\right|\leq1,\quad\forall\left[\omega_{1},\omega_{2}\right]\in\left[-\pi,\pi\right]\times\left[-\pi,\pi\right],
\]
if and only if there exist a Hermitian matrix $Q\succeq0$ such that
\begin{equation}
\begin{bmatrix}Q & h\\
h^{*} & 1
\end{bmatrix}\succeq0,\quad\delta_{k,\ell}=tr\left(\Theta_{k,\ell}Q\right),\quad k,\ell\in\mathcal{H},\label{eq:BRL}
\end{equation}
where $h$ is a column stacked vector of $\{h_{k,\ell}\}$ and $\mathcal{H}$ is a half space.
\end{lem}
Applying the Bounded Real Lemma, we can now show that the constraint
of (\ref{eq:opt1}) can be recast as the intersection of a cone
of semi-definite matrix with an affine hyperplane:
\begin{lem}
\label{Lemma:1}$\left\Vert F_{N}^{*}\alpha\right\Vert _{\infty}\leq1$
\textbf{if and only if} there exists a Hermitian matrix $Q\in\mathbb{C}^{(2N+1)^{2}\times(2N+1)^{2}}$
such that 
\begin{equation}
\begin{bmatrix}Q & \grave{h}\\
\grave{h}^{*} & 1
\end{bmatrix}\succeq0,\quad tr\left(\Theta_{k,\ell}Q\right)=\delta_{k\text{,\ensuremath{\ell}}},\quad (k,\ell)\in\mathcal{H},\label{eq:SDP}
\end{equation}
where {$\mathcal{H}$ is a half plane,} {$\grave{h}$ $\in\mathbb{C}^{\left(2N+1\right)^{2}}$
}{is a column stacked vector of $\grave{h}_{k,\ell}$
given by}
\begin{equation}
\grave{h}_{k,\ell}=\sum_{n=0}^{N}h_{n,k,\ell},\qquad h_{n,k,\ell}:=\begin{cases}
A_{n,k}\alpha_{n,k}\beta_{n,k,\ell} & k,\ell\in[-n,n],\\
0 & o.w.
\end{cases}\label{eq:h}
\end{equation}
{$\beta_{n,k,\ell}$ are given by the unique trigonometric decomposition
}{{of the associated Legendre polynomial of order $n$ and
degree $k$, i.e. $P_{n,k}\left(cos\theta\right)=\sum_{\ell=-n}^{n}\beta_{n,k,\ell}e^{j\ell\theta}$,
and $A_{n,k}$ are given in (\ref{eq:A}).}}\end{lem}
\begin{proof}
Fix a point on the two-dimensional sphere $\xi:=(\theta,\phi)\in\mathbb{S}^{2}.$
A spherical harmonic polynomial of degree $N$ is of the form 
\begin{eqnarray}
F_{N}^{*}\alpha\left(\xi\right) & = & \sum_{n=0}^{N}\sum_{k=-n}^{n}\alpha_{n,k}Y_{n,k}\left(\xi\right)\label{eq:FNa}\\
 & = & \sum_{n=0}^{N}\sum_{k=-n}^{n}\alpha_{n,k}A_{n,k}e^{jk\phi}P_{n,k}\left(\cos\theta\right).\nonumber 
\end{eqnarray}
$P_{n,k}\left(cos\theta\right)$ takes the form of 
\[
P_{n,k}\left(cos\theta\right)=\left(sin\theta\right)^{\left|k\right|}L_{n}^{\left(k\right)}\left(cos\theta\right),
\]
where $L_{n}^{\left(k\right)}$ is the $k^{th}$ derivative of the
Legendre polynomial of degree $n$. Hence, $P_{n,k}\left(cos\theta\right)$
is a trigonometric polynomial of degree $n$, and has an expansion
$P_{n,k}\left(cos\theta\right)=\sum_{\ell=-n}^{n}\beta_{n,k,\ell}e^{j\ell\theta}$
for\textbf{ }unique\textbf{ }coefficients$\left\{ \beta_{n,k,\ell}\right\} _{\ell}$.
Consequently, we write (\ref{eq:FNa}) as 
\begin{eqnarray}
F_{N}^{*}\alpha\left(\xi\right) & = & \sum_{n=0}^{N}\sum_{k=-n}^{n}\sum_{\ell=-n}^{n}\alpha_{n,k}A_{n,k}\beta_{n,k,\ell}e^{j\ell\theta}e^{jk\phi}\nonumber \\
 & = & \sum_{k=-N}^{N}\sum_{\ell=-N}^{N}\grave{h}_{k,\ell}e^{j\ell\theta}e^{jk\phi},\label{eq:Sh_tri}
\end{eqnarray}
where $\grave{h}_{k,\ell}$ is given in (\ref{eq:h}). Now, The Bounded
Real Lemma can be directly applied in our case, since the polynomial
$e^{jN\left(\theta+\phi\right)}F_{N}^{*}\alpha\left(\xi\right)$ is
causal and has the same magnitude as $F_{N}^{*}\alpha\left(\xi\right).$
This completes the proof.
\end{proof}
Using Lemma \ref{Lemma:1}, the dual problem (\ref{eq:opt1}) is equivalent
to 
\begin{equation}
max_{\alpha,Q}\left\langle y,\alpha\right\rangle \quad s.t.\quad \mbox{\emph{equation $\left(\ref{eq:SDP}\right)$ is satisfied}}.\label{eq:SDP_formula}
\end{equation}
This is a semi-definite programming optimization problem, which can
be solved using off-the-shelf software. Note that there are $\left(2N+2\right)^{4}/2$
decision variables, without any dependence on the solution accuracy. 

Define $q\left(\xi\right):=F_N^*\alpha(\xi)$, where $\alpha(\xi)$ is the solution of (\ref{eq:SDP_formula}). Denote the roots of the polynomials $1-q\left(\xi\right)$
and $1+q\left(\xi\right)$ by $\hat{\Xi}:=\left\{\hat{\xi}_m\right\}$ and recall that we know that $q(\xi)$ takes the values $sign(c_{m})$
at $\xi_{m}\in\Xi$. Consequently, $\Xi\subseteq\hat{\Xi}$. Once we find the support, we can find
the unknown coefficients by solving the least square system: 
\begin{equation}
\sum_{\hat{\xi}_{m}\in\hat{\Xi}}\hat{c}_{m}\overline{Y}_{n,k}\left(\hat{\xi}_{m}\right)=y_{n,k},\quad n=0,\dots,N,\thinspace k=-n,\dots,n.\label{eq:LS}
\end{equation}
We note that although the detected support may be larger than the actual support, the least-square solution (\ref{eq:LS}) will set the values of
the signal to zero at these points. 

The sole situation in which our algorithm fails to recover the signal
is when $q\left(\xi\right)=1$ or $q\left(\xi\right)=-1$ for all
$\xi$. However, this situation will rarely occur if (\ref{eq:SDP_formula})
is solved using standard interior point method. More precisely, according
to the analysis in Section 4 of \cite{candes2013towards}, $q\left(\xi\right)$
will not get a constant value if \emph{there exists} a solution to
(\ref{eq:SDP_formula}) such that $\vert q\left(\xi\right)\vert<1$ for some
$\xi\in S^{2}$. Indeed, in the course of our experiments this situation
never occurred.

\section{Proof of Theorem \ref{Theorem2} and Corollary \ref{Corollary1}} \label{sec:proof_noise}

Let $\hat{f}\in D_L$ be the solution of the optimization problem {(\ref{6})}, 
with $\Vert \hat{f}\Vert_{TV}=\Vert \hat{f}\Vert_{\ell_1}\leq\left\Vert f\right\Vert_{\ell_1}=\left\Vert f\right\Vert_{TV}$ 
and let $h\in D_L$, $h(\xi):=\hat{f}(\xi)-{f}(\xi)$. We decompose $h$ as 
\[
h=h_{\Xi}+h_{\Xi^{C}},
\]
where $h_{\Xi}$ and $h_{\Xi^C}$ are the parts of $h$ with support in $\Xi$ and $\Xi^C$, respectively. If $h_{\Xi}=0$,
then $h=0$. Otherwise, $h_{\Xi^{C}}\neq0$ which implies the contradiction
$\Vert \hat{f}\Vert _{\ell_1}>\Vert f\Vert _{\ell_1}$.
Using the notation $P_N:=\left(F_N\right)^*F_N$, we decompose the measure $h$ into `low' and `high' resolution parts,
\[
h_{L}=P_{N}hd\xi,\quad h_{H}=h-h_{L},
\]
where $d\xi$ is the usual surface area measure on the sphere, so that $\Vert h\Vert_{\ell_1}=\Vert h\Vert_{TV}\leq \Vert h_L\Vert_{TV}+\Vert h_H\Vert_{TV}$.

We commence by assuming that $\Vert \eta\Vert_2\leq \varepsilon$. This in turn implies that
\[
\Vert F_{N}f-y\Vert_{\ell_2}=\Vert\eta\Vert_{\ell_2}\leq\varepsilon.
\]
Using the fact that $\{Y_{n,k}\}$ is an orthobasis and then (\ref{6}) we get
\begin{align*}
\Vert P_{N}h\Vert_{L_2}&=\Vert F_{N}h\Vert_{\ell_2}\\
&\leq\Vert y-F_{N}f\Vert_{\ell_2}+\Vert y-F_{N}\hat{f}\Vert_{\ell_2} \\ & \leq 2\varepsilon.
\end{align*}
Consequently, we have the following estimation for the low resolution part $h_L$:
\[
\Vert h_L\Vert _{TV}= \Vert P_{N}h\Vert_{L_1}\leq 2\sqrt{\pi}\Vert P_{N}h\Vert_{L_2}\leq  4\sqrt{\pi}\varepsilon.
\]

Next, we need to estimate the `high frequency' part $h_H$. We denote by $h_{H,\Xi}$ and $h_{H,\Xi^{C}}$ the parts of $h_{H}$
with support on $\Xi$ and $\Xi^{C}$, respectively. By assumption, the support of $f$, $\Xi:=\left\{ \xi_{m}\right\} $
satisfies the separation condition of Definition \ref{def:separation}.
Therefore, by Lemma  \ref{Lemma:q}, there exists a polynomial $q\in V_N$ such that
$q\left(\xi_{m}\right)=sgn\left(h_{H}\left(\xi_{m}\right)\right)$ for all $\xi_m\in \Xi$ and
$\left|q(\xi)\right|<1$ for all $\xi\in \mathbb{S}_{L}^{2} \backslash\Xi$. By construction,
\[
\left\langle q,h_{H}\right\rangle =\left\langle P_{N}q,h_{H}\right\rangle =\left\langle q,P_{N}\left(h-h_{L}\right)\right\rangle =0.
\]
So,
\[
0=\left\langle q,h_{H}\right\rangle =\left\langle q,h_{H,\Xi}\right\rangle +\left\langle q,h_{H,\Xi^{C}}\right\rangle \geq\left\Vert h_{H,\Xi}\right\Vert _{TV}-q_{max}^C\left\Vert h_{H,\Xi^{C}}\right\Vert _{TV},
\]
where 
\begin{equation*}
q_{max}^C:=\max_{\xi\in\mathbb{S}_{L}^{2}\backslash\Xi}\left|q\left(\xi\right)\right|.
\end{equation*}

Since $\hat{f}$ has minimal $TV$ norm in $D_L$,
\begin{eqnarray*}
\left\Vert f\right\Vert _{TV} & \geq & \left\Vert f+h\right\Vert _{TV}\geq\left\Vert f+h_{H}\right\Vert _{TV}-\left\Vert h_{L}\right\Vert _{TV}\\
 & \geq & \left\Vert f\right\Vert _{TV}+\left\Vert h_{H,\Xi^{C}}\right\Vert _{TV}-\left\Vert h_{H,\Xi}\right\Vert _{TV}-\left\Vert h_{L}\right\Vert _{TV}\\
 & \geq & \left\Vert f\right\Vert _{TV}+\left(1-q_{max}^C\right)\left\Vert h_{H,\Xi^{C}}\right\Vert _{TV}-\left\Vert h_{L}\right\Vert _{TV}.
\end{eqnarray*}
Hence, 
\begin{eqnarray*}
\left\Vert h\right\Vert _{TV} & \leq & \left\Vert h_{L}\right\Vert _{TV}+\left\Vert h_{H}\right\Vert _{TV}\leq\left\Vert h_{L}\right\Vert _{TV}+\left\Vert h_{H,\Xi}\right\Vert _{TV}+\left\Vert h_{H,\Xi^{C}}\right\Vert _{TV}\\
 & \leq & \left\Vert h_{L}\right\Vert _{TV}+\left(1+q_{max}^C\right)\left\Vert h_{H,\Xi^{C}}\right\Vert _{TV}\\
 & \leq & \left\Vert h_{L}\right\Vert _{TV}+\frac{\left(1+q_{max}^C\right)}{\left(1-q_{max}^C\right)}\left\Vert h_{L}\right\Vert _{TV}\\
 & \leq & \frac{2\left\Vert h_{L}\right\Vert _{TV}}{\left(1-q_{max}^C\right)}\leq\frac{8\sqrt{\pi}\varepsilon}{\left(1-q_{max}^C\right)}.
\end{eqnarray*}

In order to estimate $q_{max}^C$, we make use of several results from
\cite{bendory2013exact}. Let $\xi\in\mathbb{S}_L^2\backslash\Xi$.
We first handle the case where $1/L\le d\left( {\xi ,\xi_{m} } \right)\le 
s/N$, for some $\xi_m \in \Xi$, where the constant $s>0$ is determined by Lemma 4.3 in \cite{bendory2013exact}. We provide an 
upper bound for $|q(\xi)|$, by analyzing the Taylor remainder of 
the univariate function $F\left( \theta \right):=q\left( {r\left( \theta 
\right)} \right)$, with $r\left( \theta \right):=\left( {1-\theta 
\mathord{\left/ {\vphantom {\theta {d\left( {\xi ,\xi_{m} } \right)}}} 
\right. \kern-\nulldelimiterspace} {d\left( {\xi ,\xi_{m} } \right)}} 
\right)\xi_{m} +\left( {\theta \mathord{\left/ {\vphantom {\theta {d\left( 
{\xi ,\xi_{m} } \right)}}} \right. \kern-\nulldelimiterspace} {d\left( {\xi 
,\xi_{m} } \right)}} \right)\xi $, $0\le \theta \le d\left( {\xi ,\xi_{m} 
} \right)$. By construction, $q\left( {\xi_{m} } \right)=\pm 1$, so without 
loss of generality, let us assume $F\left( 0 \right)=q\left( {\xi_{m} } 
\right)=1$. Also, by the construction in \cite{bendory2013exact}, in this case, $q$ has a local 
maximum at $\xi_{m} $ and so ${F}'\left( 0 \right)=0$. Next, by Lemma 4.3 in \cite{bendory2013exact}
there exists an absolute constant $c>0$, such that ${F}''\left( \theta 
\right)\le -cN^{2}$, for all $0\le \theta \le s/N$. Therefore, we can apply 
the Taylor Remainder theorem to bound 
\begin{eqnarray*}
 q\left( \xi \right)&=&F\left( {d\left( {\xi ,\xi_{m} } \right)} \right) \\ 
 &\le & 1-\frac{cN^{2}}{2}d\left( {\xi ,\xi_{m} } \right)^{2} \\ 
 &\le & 1-\frac{cN^{2}}{2}\frac{1}{L^{2}}=1-\frac{\tilde{{c}}}{SRF^{2}}, \\ 
\end{eqnarray*}
where SRF is defined by (\ref{eq:SRF}). We now prove the case $1/L\le s/N\le d\left( {\xi ,\xi_{m} } \right)$, 
$\forall \xi_{m} \in \Xi$. By Lemma 4.4 in \cite{bendory2013exact}, there exists 
$0<\delta <s$, such that $\left| {q\left( \xi \right)} \right|\le 
{\left( {1+\delta } \right)} \mathord{\left/ {\vphantom {{\left( 
{1+\delta } \right)} {\left( {1+s} \right)}}} \right. 
\kern-\nulldelimiterspace} {\left( {1+s} \right)}$. Therefore, if $L$ is 
chosen sufficiently large, such that 
\[
\left( {\frac{N}{L}} \right)^{2}\le \tilde{{c}}^{-1}\frac{s-\delta 
}{s+1}\quad ,
\]
then,
\[
\left| {q\left( \xi \right)} \right|\le \frac{1+\delta }{1+s}\le 
1-\frac{\tilde{{c}}}{SRF^{2}}.
\]
Applying the upper bound on $q_{max}^C$ gives
\[
\left\Vert h\right\Vert _{\ell_1} = \left\Vert h\right\Vert _{TV}\leq\frac{8\sqrt{\pi}\varepsilon}{\tilde{c}}SRF^{2}.
\]
This concludes the proof of Theorem \ref{Theorem2}

In order to prove Corollary \ref{Corollary1}, we assume that $\eta_{n,k}$ are iid entries $\mathcal{N}\sim(0,\sigma^{2})$. 
Fix $\gamma>0$ and let us denote $\varepsilon:=(N+1)(1+\gamma)\sigma$. Since $\Vert\eta\Vert_{\ell_2}^{2}$ has a $\chi^{2}$
distribution with $(N+1)^{2}$ degrees of freedom, we have (see a
comment to Lemma 1 in Section 4.1 in \cite{laurent2000adaptive})
\[
Prob\left(\Vert\eta\Vert_{\ell_2}>\varepsilon\right)\leq e^{-\frac{1}{2}(N+1)^{2}\gamma^{2}}.
\]
Therefore, 
\[
\Vert F_{N}f-y\Vert_{\ell_2}=\Vert\eta\Vert_{\ell_2}\leq\varepsilon,
\]
with probability of at least $1-e^{-\frac{1}{2}(N+1)^{2}\gamma^{2}}$.
The rest of the proof is identical to the proof of Theorem \ref{Theorem2}.

\section{Numerical Experiments\label{sec:Numerical-Experiments}}

This section is devoted to extensive numerical experiments, examining
both accuracy and complexity of Algorithm \ref{alg1}. The experiments were conducted in Matlab using CVX
\cite{cvx}, which is the standard modeling system for convex optimization.
The Matlab code is available on \cite{code}.

\begin{figure}
\begin{centering}
\includegraphics[scale=0.5]{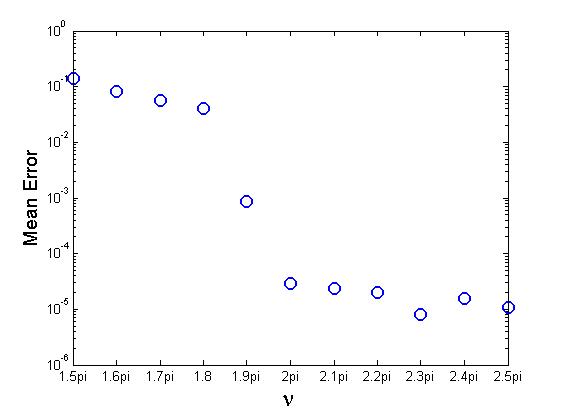} 
\par\end{centering}
\protect\caption{The mean recoery error (in logarithmic scale) as a function of $\nu$ over 20 simulations. To be clear, by error we merely mean the distance on the sphere between the true and the estimated support. \label{fig:separation}}
\end{figure}

The signals were generated in the following two stages:
\begin{itemize}
\item Random locations on the sphere were drawn uniformly, sequentially added to the signal's
support, while maintaining the separation condition of Definition \ref{def:separation}. In the non-negative
case, the support was determined by $N$ random locations (no separation
is needed). 
\item Once the support was determined, the amplitudes were drawn randomly
from an iid normal distribution with standard deviation of $SD=10$.
In the non-negative case, the amplitudes were drawn from a uniform distribution
on the interval $[0,10]$.
\end{itemize}

The first experiment aims to estimate the separation constant $\nu$ from Definition \ref{def:separation}. For each values of $\nu$,
20 simulations were conducted and the error in the support localization was calculated. As can be seen in Figure \ref{fig:separation}, starting from $\nu=2\pi$ the localization error is negligible. The result suggests that there exists a sharp phase transition for the relationship between the recovery error and $\nu$ (see \cite{moitra2014threshold} for the analysis of this phenomenon for signals defined on the circle). 
In the presence of noise, we found that increasing the separation constant to $\nu=2.5\pi$ improves the results significantly.   

An example to the performance of the algorithm is presented in Figure
\ref{fig:result1}. Figure \ref{fig:result1a} presents the low resolution
measurements $P_{N}f:=F_{N}^{*}F_{N}f$ for $N=10$, and the recovered
signal is presented in Figure \ref{fig:result1b}%
\footnote{The signal is presented on a grid for visualization purpose only.%
}. We note that the recovered signal is identical to any visible accuracy to the original signal. As mentioned in Algorithm \ref{alg1}, the support is determined
as the roots of the polynomials $1\pm q(\xi)=1\pm F_{N}^{*}\alpha(\xi),$ where
$\alpha$ is the solution of the semi-definite program (\ref{eq:SDP_formula}).
Figure \ref{fig:q} presents $q\left(\xi\right)$ for the example
of Figure \ref{fig:result1}. 

The roots of $1\pm q\left(\xi\right)$ were located as follows. The
sphere was divided into small cubes, and the minimum of the function
at each cube was calculated using standard optimization tools. The
minimum points with values below $10^{-5}$ were considered
as roots. This technique exploits the fact that the signal's support
is well separated. Finer segmentation of the sphere results in a better
localization in the cost of computation time.

This experiment was conducted 10 times for $N=5,8,10$. Table \ref{tab:1}
shows the error in estimating the support locations. Figure \ref{fig:non-negative}
shows an exact recovery of a clustered non-negative signal. As aforementioned,
the separation condition is not necessary in this case. 

\begin{figure}
\begin{minipage}[t]{0.45\columnwidth}%
\subfloat[The low resolution measurements $P_{N}f$, for $N=10.$\label{fig:result1a}]{
\includegraphics[scale=0.4]{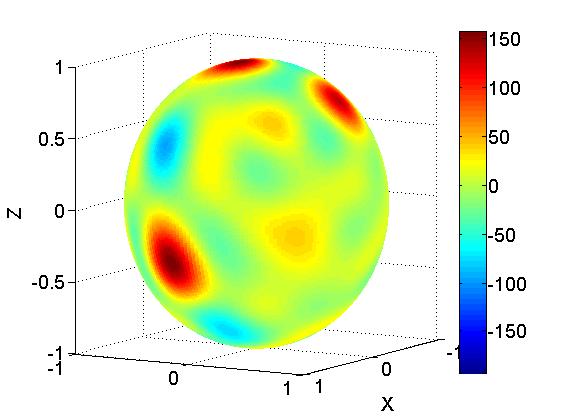}

}%
\end{minipage}%
\begin{minipage}[t]{0.45\columnwidth}%
\subfloat[The recovered signal $f$, for $N=10.$\label{fig:result1b}]{
\includegraphics[scale=0.4]{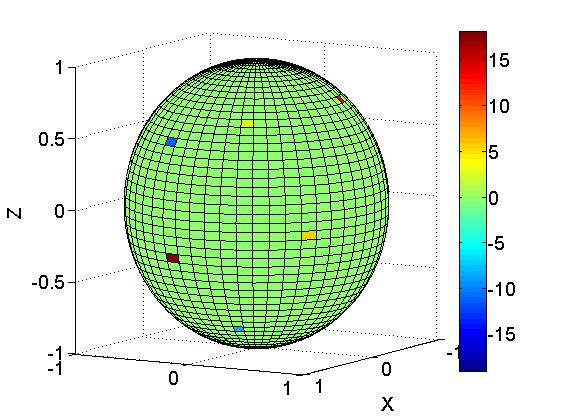}

}%
\end{minipage}

\begin{minipage}[t]{0.45\columnwidth}%
\subfloat[The low resolution measurements $P_{N}f$, presented on a plane.]{
\includegraphics[scale=0.4]{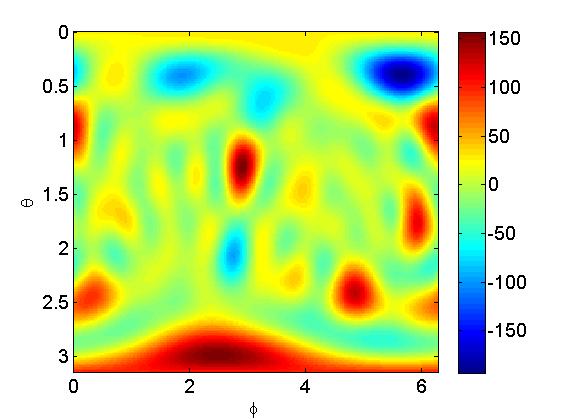}

}%
\end{minipage}%
\begin{minipage}[t]{0.45\columnwidth}%
\subfloat[The recovered signal $f$, presented on a plane.]{
\includegraphics[scale=0.4]{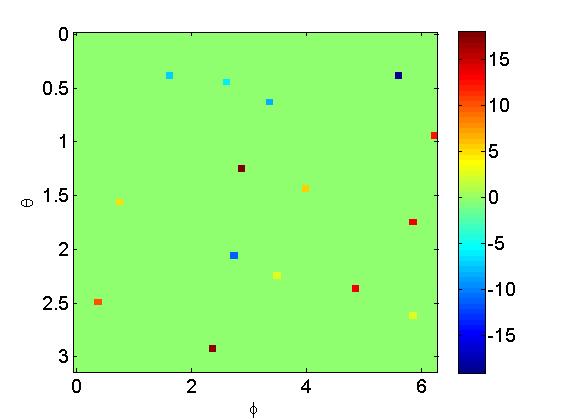}

}%
\end{minipage}

\centering{}\protect\caption{Super-resolution on the sphere using the Algorithm \ref{alg1}, for $N=10$.
The signal is presented on a grid for visualization only.\label{fig:result1}}
\end{figure}

\begin{center}
\begin{figure}
\begin{centering}
\begin{minipage}[t]{0.45\columnwidth}%
\subfloat[The function $q(\xi)=F_{N}^{*}\alpha(\xi)$.]{
\includegraphics[scale=0.4]{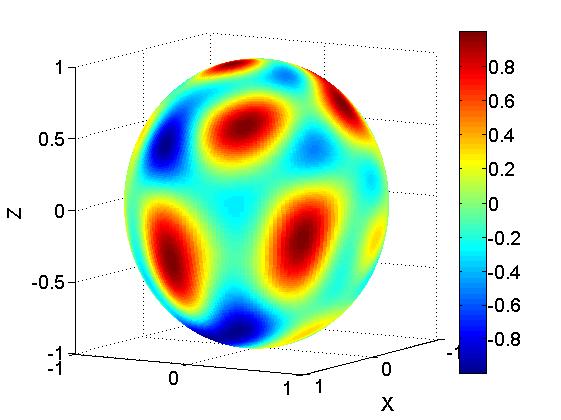}

}%
\end{minipage}\hfill{}%
\begin{minipage}[t]{0.45\columnwidth}%
\subfloat[A single line of $q(\xi)$ for a constant $\theta=2.257$ (blue) verses
the appropriate values of $sign(f)$ (red).]{
\includegraphics[scale=0.4]{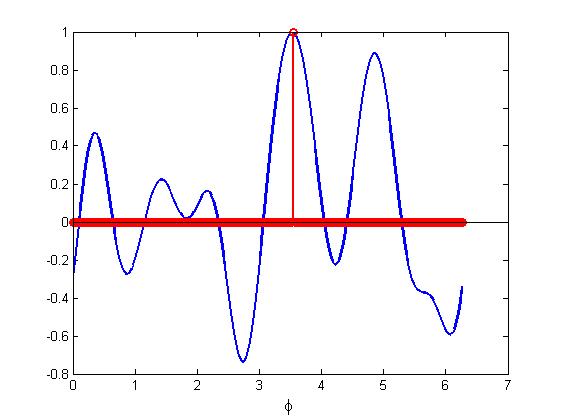}

}%
\end{minipage}
\par\end{centering}

\centering{}\protect\caption{The function $q(\xi)=F_{N}^{*}\alpha(\xi)$ for the example presented
in Figure \ref{fig:result1}. \label{fig:q} }
\end{figure}

\par\end{center}

\begin{table}
\begin{centering}
\begin{tabular}{|c|c|c|c|}
\hline 
\multicolumn{1}{|c|}{N} & 5 & 8 & 10\tabularnewline
\hline 
\hline 
Average error & $8.1267\times10^{-5}$ & $8.1826\times10^{-5}$ & $9.0404\times10^{-5}$\tabularnewline
\hline 
Max error & $2.163\times10^{-4}$ & $1.9\times10^{-3}$ & $3.3\times10^{-3}$\tabularnewline
\hline 
\end{tabular}
\par\end{centering}

\protect\caption{The localization error of Algorithm \ref{alg1} for $N=5,8,10.$ For
each value of $N$, the experiment was conducted 10 times. \label{tab:1}}
\end{table}

\begin{figure}
\begin{centering}
\begin{minipage}[h]{0.32\columnwidth}%
\subfloat[The original signal $f$ with 9 non-zero values.]{
\includegraphics[scale=0.3]{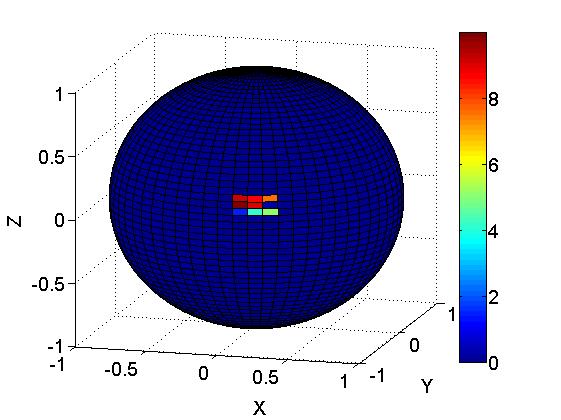}

}%
\end{minipage}%
\begin{minipage}[h]{0.32\columnwidth}%
\subfloat[The signal projection $P_{N}f$ for $N=9$. ]{
\includegraphics[scale=0.3]{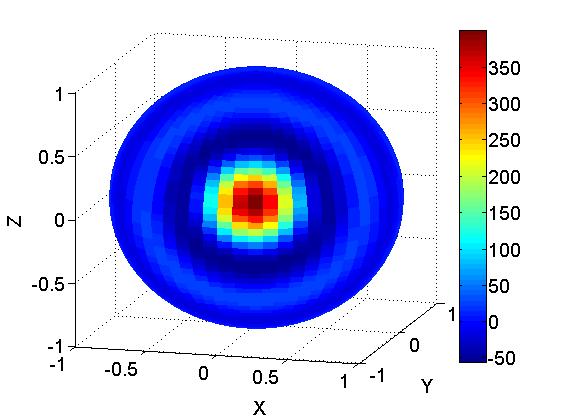}

}%
\end{minipage}%
\begin{minipage}[h]{0.32\columnwidth}%
\noindent \subfloat[The recovered signal.]{\noindent
 \includegraphics[scale=0.3]{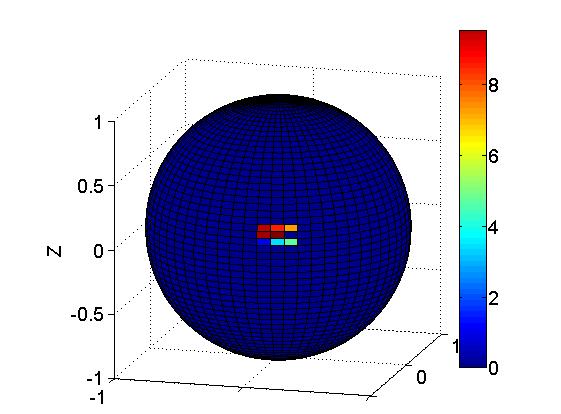}

}%
\end{minipage}
\par\end{centering}
\protect\caption{Recovery of a clustered non-negative signal. The signals are presented
on a grid for convenient visualization. \label{fig:non-negative}}
\end{figure}

 In the discrete setting, both Algorithm \ref{alg1} and $\ell_1$ minimization recover the signal exactly in a noise-free environment. In order to compare the algorithms, we applied both of them in the discrete setting, using the grid
\[
\mathbb{S}_{L}^{2}:=\left\{ \left(\phi_{q},\theta_{p}\right)=\left(2\pi\frac{q}{L},\pi\frac{p}{L}\right)\thinspace:\thinspace\left(q,p\right)\subset\left[0,1,\dots,L-1\right]\times\left[0,1,\dots,L-1\right]\right\} .
\]
Recall that the complexity of Algorithm \ref{alg1}  depends solely on
$N$, while the complexity of the $\ell_{1}$ minimization depends
on both $N$ and $L$. Therefore, finer grid results in a longer computation
time. 
Figure \ref{fig:discrete} shows the computation time of $\ell_{1}$ minimization, as function
of the SRF, compared with the average computation time of Algorithm \ref{alg1}. The $\ell_{1}$ minimization was solved using CVX \cite{cvx}.
As can be seen, the computation time of $\ell_{1}$ minimization grows
(approximately) linearly with the SRF. In the last section we discuss some ideas how to speed up our algorithm.

\begin{figure}
\begin{centering}
\includegraphics[scale=0.5]{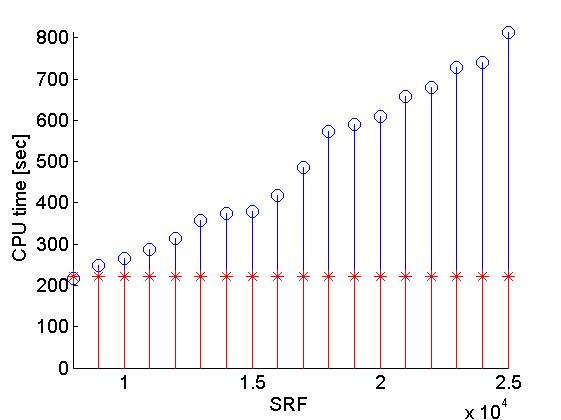}
\par\end{centering}

\protect\caption{CPU time of $\ell_1$ minimization and Algorithm \ref{alg1} in the discrete setting with
$N=5$. The red stars present the average computation time of Algorithm \ref{alg1} over 100
experiments, and the blue circles present the average computation
time (over 10 experiments) of the $\ell_{1}$ minimization as a function
of the SRF.\label{fig:discrete}}
\end{figure}

In the noisy setting, we considered an additive noise with  iid entries $\mathcal{N}\sim(0,\sigma^{2})$.
Our experiments show that although Theorem \ref{Theorem2} holds  for recovery by $\ell_1$ minimization (\ref{6}), Algorithm \ref{alg1} is also robust to noise. In Figure \ref{fig:result_noise}, we show the recovery error of Algorithm \ref{alg1} with $\varepsilon=(1+N)(1+\gamma)\sigma$ for various of $\sigma$ and SRF values in the discrete setting. As can be seen, the error degrades moderately as $\sigma$ and SRF increase. 

\begin{figure}
\begin{minipage}[t]{0.45\columnwidth}%
\subfloat[The recovery error as a function of the noise standard deviation, for $SRF=10$.   \label{fig:noisea}]{\includegraphics[scale=0.4]{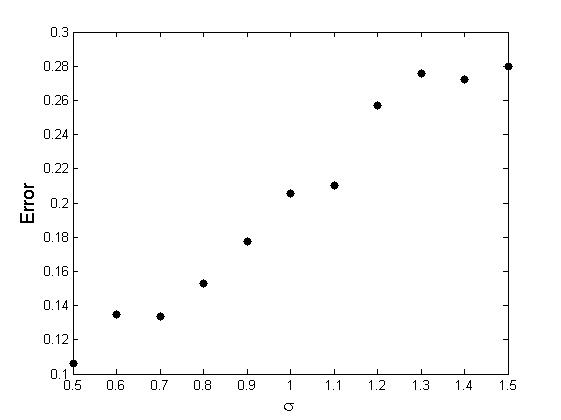}

}%
\end{minipage}%
\begin{minipage}[t]{0.45\columnwidth}%
\subfloat[The recovery error as a function of the SRF, for $\sigma=0.5$. \label{fig:result1b-1}]{\includegraphics[scale=0.4]{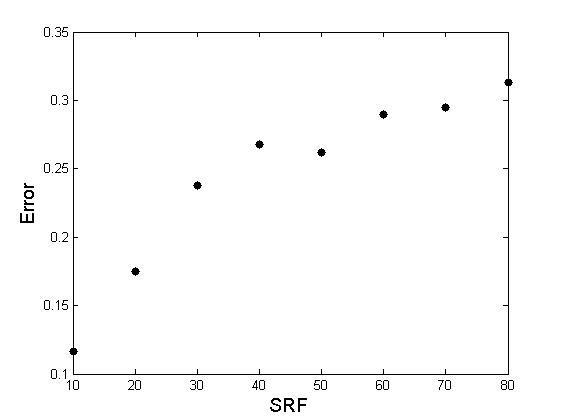}

}%
\end{minipage}

\centering{}\protect\caption{For each value of $\sigma$ and SRF, 10 experiments were conducted using Algorithm \ref{alg1} with $N=8$, $\gamma=1$, and $\varepsilon=(1+N)(1+\gamma)\sigma$. The figure presents the average recovery error. 
By error, we merely mean the distance on the sphere between the true and the estimated supports. \label{fig:result_noise}}
\end{figure}

\section{ the necessity of the separation condition} \label{sec:necessity}

In \cite{bendory2013exact}, we established that the separation condition is a \emph{sufficient}
condition for exact recovery of a signal on a sphere from its projection
onto $V_{N}$. Nonetheless, without separation the recovery task is
ill-posed. Figure \ref{fig:cluster} shows an example for Algorithm \ref{alg1} failure for clustered signals. The experiment was conducted
with minimal separation of $\frac{2.5\pi}{N^{1.5}}$, and $N=10.$
The points are scattered on $1/N$ of the sphere, so the total number
of locations is similar to the experiment presented in Figure \ref{fig:result1}.

\begin{figure}
\begin{minipage}[t]{0.45\columnwidth}%
\subfloat[Clustered signal with separation of $\frac{2.5\pi}{N^{1.5}}$. ]{
\includegraphics[scale=0.4]{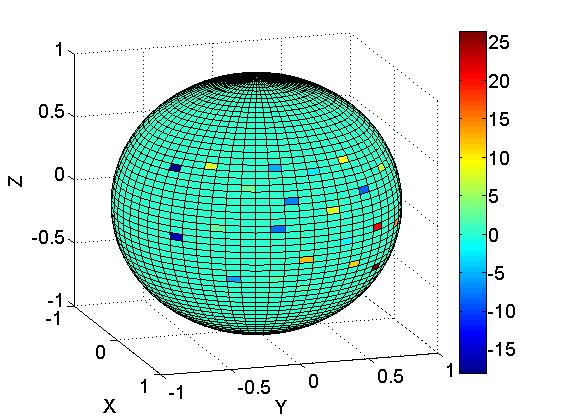}

}%
\end{minipage}%
\begin{minipage}[t]{0.45\columnwidth}%
\subfloat[The recovered signal.]{
\includegraphics[scale=0.4]{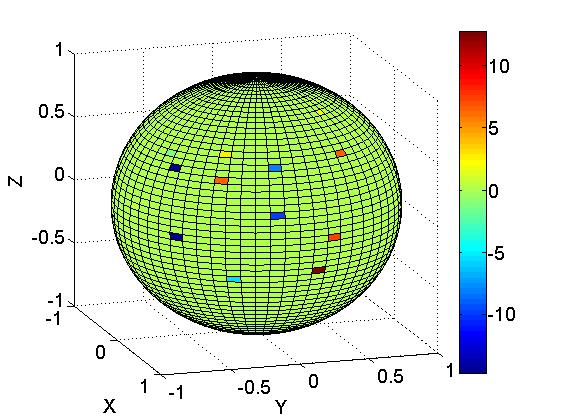}

}%
\end{minipage}

\protect\caption{Unsuccessful recovery through Algorithm \ref{alg1} of a clustered
signal, $N=10$. \label{fig:cluster}}
\end{figure}

The recovery failure of the clustered signal can be explained by the analysis
in \cite{candes2013towards}, where the authors showed that clustered
signals cannot be recovered \emph{by any method} from their low frequency
coefficients in the presence of minuscule noise level. They used prolate
spheroidal sequences \cite{slepian1978prolate}, and showed that asymptotically,
even for small $SRF$ values, there will always exist an irretrievable
signal. Furthermore, as the $SRF$ increases most of the information
in a clustered signal is lost by the low-pass operation. 

The formulation (\ref{eq:Sh_tri}) reveals that spherical harmonics
expansion in $\mathbb{S}^{2}$ is a unique combination of bivariate
trigonometric polynomials. Hence, the aforementioned conclusions hold
for the spherical harmonic case as well, and we conclude that super-resolution
on the sphere is ill-posed without a minimum separation condition. 

To make this argument clear, we give here a simple example. Consider
a signal of the form 
\[
f=\delta_{\xi}-\delta_{\xi_{\epsilon}},
\]
for some locations $\xi:=(\theta,\phi),$ $\xi_{\epsilon}:=(\phi_{\epsilon},\theta_{\epsilon})$,
and $d\left(\xi,\xi_{\epsilon}\right)=\epsilon$ . In this case, the
measurements are given by 
\begin{eqnarray*}
y_{n,k} & = & \overline{Y}_{n,k}\left(\xi\right)-\overline{Y}_{n,k}\left(\xi+\epsilon\right)\\
 & = & A_{n,k}\left[e^{-jk\phi}P_{n,k}\left(\cos\theta\right)-e^{-jk\phi_{\epsilon}}P_{n,k}\left(\cos\theta_{\epsilon}\right)\right]\\
 & = & A_{n,k}\sum_{\ell=-n}^{n}\beta_{n,k,l}\left[e^{-jk\phi}e^{-j\ell\theta}-e^{-jk\phi_{\epsilon}}e^{-j\ell\theta_{\epsilon}}\right].
\end{eqnarray*}
When the spikes are close, i.e. $\epsilon\rightarrow0$, we get $\left|y_{n,k}\right|\rightarrow0$
\emph{for any} $\left(n,k\right).$ As a conclusion, if the spikes
are sufficiently close, all the measured information will be completely
drowned in a small noise level.

\section{Conclusions and Future Work \label{sec:Conclusions-and-Future}}

In a previous paper, we have established that a signal of the form
(\ref{eq:signal}) on the sphere can be recovered precisely from its
 projection onto spherical harmonics of degree $N$ using TV minimization, as long as the
distance between the spikes is inversely proportional to $N$. In
this paper, we suggested to recast the infinite dimensional TV minimization
as a semi-definite program with $\mathcal{O}(N^{4})$ variables. We
showed that Algorithm \ref{alg1} recovers the signal with high precision and that its complexity does not depend on the resolution.
 We strongly believe that this result holds in higher dimensions and for complex Dirac ensembles. Indeed, significant
parts of the proof can be easily generalized to any dimension and to complex signals. However, there
are certain technical challenges which we hope to overcome in
future work.

Furthermore, we showed that in the discrete configuration, recovery by  $\ell_1$ minimization is robust to noise, and the recovery error is proportional to the noise standard deviation and $SRF^2$. We showed experimentally that similar estimation holds for Algorithm \ref{alg1} as well.

Our algorithm can be improved in two directions. As aforementioned,
the semi-definite program was implemented using CVX on Matlab. CVX was designed
as a convenient tool for convex optimization, however, it does not
purport to be efficient. In order to speed the computation time, one
needs to design its own solver, which is beyond the scope of this paper.
Another direction is designing an algorithm, dedicated for root finding
on the sphere. 

Finally, this paper is part of an ongoing research, suggesting to
solve infinite dimensional convex optimization problems using a finite
semi-definite programs. Up to now, these method were applied to projections
onto trigonometric \cite{candes2013towards,candes2013super,bhaskar2011atomic,tang2012compressive,tang2013near,xu2013precise,chi2013compressive}
and algebraic polynomial spaces \cite{bendory2013Legendre,de2014non}
. This work showed that it can be applied to the sphere as well. An
interesting question is whether this approach can be applied to recover
a signal lying on complicate geometries from their projection onto harmonic
polynomials. We leave this for a future research. 

\section*{Acknowledgement}
We thank the referees for their valuable comments that
have significantly improved this paper.

\bibliographystyle{plain}
\bibliography{bib}

\end{document}